\documentclass[11pt]{article}
\usepackage{fullpage}
\usepackage{amsmath,amsthm,amssymb}
\usepackage{xcolor}
\usepackage{algpseudocode,algorithm,algorithmicx}
\usepackage{mathtools}
\usepackage{sectsty}
\usepackage{hyperref}
\usepackage[nameinlink,capitalise]{cleveref}
\usepackage{thm-restate}
\usepackage{natbib}
\usepackage{bm}
\usepackage{bbm}
\usepackage{apptools}
\newcommand*\samethanks[1][\value{footnote}]{\footnotemark[#1]}
\usepackage{caption,subcaption}

\newtheorem{theorem}{Theorem}[section]
\newtheorem{lemma}[theorem]{Lemma}
\newtheorem*{lemma*}{Lemma}
\newtheorem{definition}[theorem]{Definition}
\newtheorem*{definition*}{Definition}
\newtheorem{proposition}[theorem]{Proposition}

\usepackage{amsmath,amsfonts,bm}

\def\eqref#1{equation~\ref{#1}}
\def\1{\bm{1}}

\def\eps{{\epsilon}}

\def\vx{{\bm{x}}}
\def\vy{{\bm{y}}}
\def\vz{{\bm{z}}}

\DeclareMathAlphabet{\mathsfit}{\encodingdefault}{\sfdefault}{m}{sl}
\SetMathAlphabet{\mathsfit}{bold}{\encodingdefault}{\sfdefault}{bx}{n}

\newcommand{\cdd}{{\alpha}} % Doubling Dimension Notation
\title{First Passage Percolation with Queried Hints\thanks{A preliminary version of this paper appeared in AISTATS 2024. Code for the experiments can be found here: \url{ https://github.com/google-research/google-research/tree/master/probe_routing}}}
\author{Kritkorn Karntikoon \thanks{Department of Computer Science, Princeton University, Princeton, NJ 08540. Email: \href{mailto:kritkorn@princeton.edu}{\texttt{kritkorn@princeton.edu}}} \and Yiheng Shen\thanks{Department of Computer Science, Duke University, Durham, NC 27708-0129. Email:
\href{mailto:yiheng.shen@duke.edu}{\texttt{yiheng.shen@duke.edu}}} \and Sreenivas Gollapudi\thanks{Google Research, Mountain View, CA 94043. Emails: \href{mailto:sgollapu@google.com}{\texttt{sgollapu@google.com}}, \href{mailto:kostaskollias@google.com}{\texttt{kostaskollias@google.com}}, \href{mailto:aschild@google.com}{\texttt{aschild@google.com}}, \href{mailto:asinop@google.com}{\texttt{asinop@google.com}}. This work was done when Kritkorn Karntikoon and Yiheng Shen were student researchers at Google
(Summer 2022 and Summer 2023)}\and Kostas Kollias\samethanks[4]\and Aaron Schild\samethanks[4]\and Ali Sinop\samethanks[4]}

\date{}

\begin{document}
\maketitle

%\footnote{This work was conducted while the authors Kritkorn Karntikoon and Yiheng Shen were student researchers at Google.}

\begin{abstract}
    Solving optimization problems leads to elegant and practical solutions in a wide variety of real-world applications. In many of those real-world applications, some of the information required to specify the relevant optimization problem is noisy, uncertain, and expensive to obtain. In this work, we study how much of that information needs to be queried in order to obtain an approximately optimal solution to the relevant problem. In particular, we focus on the shortest path problem in graphs with dynamic edge costs. We adopt the {\em first passage percolation} model from probability theory wherein a graph $G'$ is derived from a weighted base graph $G$ by multiplying each edge weight by an independently chosen random number in $[1, \rho]$. Mathematicians have studied this model extensively when $G$ is a $d$-dimensional grid graph, but the behavior of shortest paths in this model is still poorly understood in general graphs. We make progress in this direction for a class of graphs that resemble real-world road networks. Specifically, we prove that if $G$ has a constant continuous doubling dimension, then for a given $s-t$ pair, we only need to probe the weights on $((\rho \log n )/ \epsilon)^{O(1)}$ edges in $G'$ in order to obtain a $(1 + \epsilon)$-approximation to the $s-t$ distance in $G'$. We also generalize the result to a correlated setting and demonstrate experimentally that probing improves accuracy in estimating $s-t$ distances.
\end{abstract}

% \begin{itemize}

%     There has been a long line of work in \item Finding shortest paths in road networks of critical importance
%     \item Long history
%     \item Lots of work in the 2000s in the theory literature
%     \item Heuristics for exact shortest path found but not analyzed
%     \item Highway dimension paper showed that they all have low complexity
%     \item Not many of them deal with traffic
%     \begin{itemize}
%         \item CRP developed to handle changing edge weights
%         \item Requires reprocessing when edge weights change
%     \end{itemize}
%     \item CRP ostensibly needs to rebuild all shortcuts, as all edge weights change
%     \item Question: Which edge weight changes actually matter? How much real traffic information do we actually need in order to find a shortest traffic-aware path?
%     \item Intuitively, when people route, they only think about a few bottlenecks, like bridges, 280 vs 101, etc. They don't need to know all edge weights.
%     \item Related to alternate routes, but no theoretical guarantees here. \item Can we theoretically show that we only need to probe a small number of real traffic values to find a path?
%     \item Yes!
%     \item State theorems
% \end{itemize}

\section{Introduction}

Path search in dynamic systems such as traffic networks is a foundational problem in computer science. Dijkstra's algorithm does not suffice for path search in large graphs, like continent-scale road networks. Work began on efficient path search in the early 2000s, with the development of reach \citep{gutman2004reach, goldberg2006reach}, contraction hierarchies \citep{geisberger2008contraction}, and more \citep{bast2006transit, bast2007transit, bauer2010sharc, goldberg2005computing, hilger2009fast, ich2006extremely, sanders2005highway}. Theoretical justification for the efficiency of these methods came shortly after with the introduction of \emph{highway dimension} \citep{AbrahamFG10}. These techniques have seen widespread adoption in routing engines.

Efficiency is not the only requirement for a useful path search engine. Users want customized routes that adapt to real-world conditions. In particular, path search engines need to find the shortest path subject to current real-world traffic and road-closure conditions. This requires designing a path search engine that can handle edge weight modifications. In the 2010s, customizable route planning (CRP) \citep{delling2011customizable} was developed to handle changing edge weights.

All of these techniques involve redoing expensive preprocessing every time the graph changes. For instance, CRP starts with a partition of the input graph and precomputes shortcuts between boundary nodes of each cluster. When traffic conditions change, CRP needs to rebuild all shortcuts associated with clusters in which an edge weight changed. Traffic conditions generally change throughout the input road network, necessitating recomputation of almost all shortcuts.

One may wonder, though, when recomputation needs to be done. In particular, if the traffic conditions on a short surface street change, only routes with nearby origin and destination are likely to be affected, so no recomputation of shortcuts should be necessary. In this work, we show both theoretically and experimentally that this recomputation is indeed unnecessary for most edge weight changes. We show, under random traffic, only a small number of real traffic values need to be queried in order to obtain a good approximation to the origin-destination distance with traffic (\cref{thm:main}).

\iffalse
\begin{enumerate}
\item Under random traffic, we prove that only a small number of real traffic values need to be queried in order to obtain a good approximation to the origin-destination distance with traffic (\cref{thm:main}).
\item Under random traffic, we observe experimentally that no edges need to be queried in order to obtain a near optimal origin-destination path.
\end{enumerate}
\fi

\subsection{Theoretical Result: Independent Models}

We now briefly discuss our theoretical results, leaving formal definitions to \cref{sec:algorithm}. In both our theoretical and experimental results, we always use random noise to model traffic, though not always noise that is independent across all edges. In our theoretical results, the graph is always undirected, though in experiments the graphs are directed. We always start with a weighted graph (e.g. road network) $G$, with weights intuitively representing free-flow traffic values. To generate traffic, multiply each edge weight by a random number in $[1,\rho]$ to produce a graph $G'$, representing a road network with traffic. Given an origin-destination pair $s-t$, our goal is to compute an $s-t$ path that is as short as possible within $G'$ without actually querying many edges in $G'$. This, in the application, would amount to computing an approximate traffic-aware shortest path without much real traffic information. Specifically, given a bound on the number of probes required, we can obtain a traffic-aware distance data structure from a no-traffic one when edge weights are independently chosen:

\begin{restatable}[]{theorem}{thmDataStructure}
\label{thm:ds}
Suppose that \cref{alg:main_alg} probes at most $K_0$ edges and that, given an $\delta > 0$ and a graph $G$, there exists a data structure $\mathcal{X}$ that takes $K_1$ time to initialize and is equipped with a method \texttt{ApproxNoTrafficDistance}($s,t,G$) that outputs a $(1 + \delta)$-approximation to $d_G(s,t)$ in $K_2$ time. Then, given $\epsilon > 0$, an $m$-edge graph $G$, a (traffic) weight distribution $\mathcal{D}$ in which edges are independent, and query access to a hidden graph $G'$ as described previously, there exists a data structure $\mathcal{Y}$ that takes $\tilde{O}(m + K_1)$ time to initialize and is equipped with a method \texttt{ApproxTrafficDistance}($s,t,\mathcal{D},G,G'$) that computes a $(1 + \epsilon)$-approximation to $d_{G'}(s,t)$ with probability at least $1 - 1/m^8$ in $\tilde{O}(K_0^2K_2)$ time.
\end{restatable}

We prove this result in \cref{app:data_structure}. As per \citep{AbrahamFG10}, $K_1 = \tilde{O}(m)$ and $K_2 = \text{polylog}(n)$ in practice. Therefore, if we can bound the number of probes $K_0$, then we obtain a data structure for traffic routing with polylogarithmic runtime per query. Thus, we focus on the problem of showing that few probes to $G'$ are required to compute shortest paths in $G'$. Doing this for arbitrary graphs $G'$ is impossible, as in graphs $G'$ with a large number of parallel edges, one has to query all of the parallel edges in order to find an approximately shortest $s-t$ path. However, this behavior is not realistic in practice, as it would require the existence of lots of potentially optimal disjoint paths between a given origin and destination. This is implausible in road networks, as most shortest paths use highways at some point and there are not many highways to choose from. This observation motivated the work of \citep{AbrahamFG10} and led them to only consider the class of graphs that have what they called low \emph{highway dimension}\footnote{For the formal definition of (continuous) highway dimension, see \cref{app:highway}.}. Since we only care (and can) find approximately shortest paths, we can show results for a broader class of graphs -- graphs with low \emph{continuous doubling dimension (cdd)}.

A graph has doubling dimension $\cdd$ if the ball of radius $2R$ around any vertex $v$ can be written as the union of balls of radius $R$ around a collection of at most $2^\cdd$ other vertices. This definition is motivated by what happens in $\mathbb{R}^\cdd$, as any $\cdd$-dimensional ball of radius $2R$ can be covered with at most $6^\cdd$ balls of radius $R$. Road networks are often assumed to have low doubling dimension \citep{abraham2006routing, abraham2016highway, marx2018}. For our purposes, we need a slightly stronger property: we need the the graph where each edge is chopped up into infinitesimally small segments to have low doubling dimension as well. This is called \emph{continuous doubling dimension (cdd)}. Our main result shows that if $G$ has low cdd, only a small number of edges need to be queried:

\begin{restatable}[]{theorem}{thmMain}
\label{thm:main}
If the graph $G$ has cdd $\cdd$, then \cref{alg:main_alg} probes at most $((\rho \log n) / \epsilon)^{O(\cdd)}$
edges, and returns a number, $\hat{\delta}$, that satisfies
$|\hat{\delta} - \delta| \leq \epsilon \delta$ with high probability. Here
$\delta$ is the actual shortest path length from $s$ to $t$ in $G'$.
\end{restatable}

Theorem \ref{thm:main} states that the number of probes required to estimate the \emph{length} of the shortest path in $G'$ is small. In practice, we of course want to produce a \emph{path} from $s$ to $t$ with short length in $G'$, not just know the length of that path. This unfortunately requires a large number of probes in general:

\begin{restatable}[]{theorem}{thmLB}
\label{thm:lb}
Any adaptive probing strategy with query complexity at most $n/100$\IfAppendix{}{\footnote{The constant 100 can be generalized to any $c$. We chose 100 for simplicity and purpose of our results.}} returns a path $P$ in $G'$ with quality $q(P) > 9/8 - 1/10 > 1$ with probability at least $1 - 2n^{-100}$ 
where $q(P)$ is the ratio between length of path $P$ and the shortest path in $G'$; i.e., $q(P) := \ell_{G'}(P) / d_{G'}(s,t)$.
\end{restatable}

Luckily, the example is somewhat pathological. Intuitively, this example enables getting on and off of a highway many times in a row, which does not make sense for any near-shortest path in real road networks. The following assumption captures this idea:

\begin{restatable}[Polynomial Paths]{assumption}{asPoly}
    \label{as:poly-paths}
    $G$ has the property that, for any origin-destination pair $s-t$, the number of distinct possible $s-t$ shortest paths (over choices of each edge's random multiplier) is at most $U = \text{poly}(|V(G)|)$ 
\end{restatable}

We show that we can in fact find a short path in $G'$ if $G$ satisfies the polynomial paths assumption, which it likely does in practice:

\begin{restatable}[]{theorem}{thmPathAlgo}
    \label{thm:path-algo}
    \cref{alg:main_alg}, with the threshold $c\epsilon^2 L/(\rho^4 \log n)$ replaced with $c\epsilon^2 L/(\rho^4 \log nU)$, finds a $(1 + \epsilon)$-approximate shortest path in $G'$ with probability at least $1 - n^{-100}$ and queries at most $((\rho \log nU)/\epsilon)^{O(\cdd)}$ edges if both of the following hold:

    \begin{enumerate}
        \item $G$ satisfies the polynomial paths assumption.
        \item $G$ has cdd at most $\cdd$.
    \end{enumerate}
\end{restatable}
The proofs of \cref{thm:lb,thm:path-algo} can be found in \cref{app:find_hard,app:get_around}.

\subsection{Theoretical Results: Correlated Model}\label{subsec:cor}

Up to this point, randomness for different edges has always been independent. This is unreasonable in practice, as traffic on nearby segments of a highway is likely to be very correlated. We model this correlation as follows.

We start with a weighted graph $G=(V, E, w)$. Suppose that there are $m$ hidden variables $\vy=(y_1, y_2, \ldots, y_m)$. Each hidden variable $y_i$ follows an independent distribution $\mathcal{D}_i$ within the range $[\lambda_i, \lambda_i \cdot \rho]$. To generate real-time traffic, each edge $e\in E$ has actual weight in the following form:
\begin{equation} \label{eqn:demand_constraint}
   w_e'(\vy) = \left(\sum_{i=1}^{m} \lambda_{i}^e \cdot y_i\right)^\beta \cdot w_e, 
\end{equation}
where $\{\lambda_i^e\}_{i \in [m], e \in E}$ are known non-negative real numbers, called \emph{dependence parameters}. These parameters indicates the influence of hidden variable $y_i$ on the edge $e$.

Think about this model as follows. Each $i$ represents a specific origin-destination pair that users may travel on, with aggregation allowed between long distance pairs. For instance, one $i$ could represent the centers of two major cities, or the centers of two suburbs. $y_i$ represents the number of people traveling between the origin and destination per hour at a randomly chosen time. Thus, $y_i$ is random. Furthermore, conditioned on the chosen time of day, different $y_i$'s are likely to be independent random variables. Each edge has total traffic as a linear combination (i.e. $\sum_{i=1}^{m} \lambda_{i}^e \cdot y_i$) of the all the demands that affect it. $f(t)=t^\beta$ is a power function mapping the total traffic on an edge to its traversing time, where $\beta$ is generally set to $4$ in the literature \citep{manual1964bureau,ccolak2016understanding,benita2020data}.

Instead of probing on edges, in this model we are allowed to probe on the demands (i.e. $y_i$'s). We use the term ``under basic demands'' to represent the scenario when all the hidden variables are at their lowest possible values, i.e. $\vy = (\lambda_1, \lambda_2,\ldots, \lambda_m)$.

For each hidden variable $y_i$, define cluster $i$ as $C_i = \{e: \lambda_i^e > 0\}$, i.e. the set of edges that are actually influenced by $y_i$. In this model, we apply the same idea as \cref{alg:main_alg} to probe the demands with the largest cluster sizes, after a normalization step. The full algorithm can be found in \cref{app:alg}. 

We have the following theorem which bounds the number of demand probes in order to estimate the shortest path length:
\begin{restatable}[]{theorem}{thmHighwayDemandsAssumption}
    \label{thm:highway-demands}
    Given a weighted graph $G = (V, E, w)$ and a pair of vertices $s$ and $t$, suppose there are some hidden random variables $\{y_i\}_{i \in [m]}$, each follows an independent bounded distribution within the range $[\lambda_i,\ \rho \cdot \lambda_i]$. Suppose that $G$ satisfies the following:
    \begin{itemize}
        \item [1.] $G$ has continuous highway dimension $h$ under basic demands;
        \item [2.] Every $C_i$ represents a shortest path under basic demands;
        \item [3.] Every edge in $G$ only falls in at most $\ell$ different clusters; 
    \end{itemize}
   In the actual graph $G' = (V, E, w')$, $w'$ satisfies \cref{eqn:demand_constraint} with dependence parameters $\{\lambda_i^e\}_{i \in [m], e \in E}$.
    Then, \cref{alg:demand_alg} uses at most
    $\left(\frac{\log n\cdot \rho^{\beta} \cdot \ell}{\eps^2}\right)^{O(\log h)}$ probes on the demands and returns a $(1+\epsilon)$-approximate shortest path length from $s$ to $t$ in $G'$ with high probability.
\end{restatable}

If one wants an exact constant dependency on $\log h$ in the exponent, it is safe to replace $O(\log h)$ with $6 \log h$. When $n\gg h$, the constant can be close to $4$.
In practice, one should think of $h$ as being constant, as road networks are known to have low highway dimension as discussed earlier. $\beta = 4$ in practice as discussed earlier. $\ell$ is likely constant in practice due to the fact that one can aggregate demands between distant locations, thus resulting in a small number of truly distinct paths that traverse a given highway segment. Thus, the number of queries stated in this theorem is polylogarithmic in practice. We will analyze the algorithm and present some ideas of proving the above theorem in \cref{sec:demand_analysis}. Apart from our main result, \cref{thm:main}, the detailed proofs of all theorems and lemmas are provided in the Appendix. 

\subsection{Experimental Results}

Our algorithms work by querying edges or paths with high edge weight in $G$. We show that querying edges does in fact improve the ability to find an approximately shortest path in $G'$. We illustrate this using several regional road networks obtained from Open Street Maps (OSM) data \citep{OpenStreetMap}.

\subsection{Related Work}

Our problem is closely related to \textit{first passage percolation} which is a classic problem in probability theory. Given a graph $G$ with random edge weights drawn from independent distributions,
the goal is to understand the behavior of the $s$-$t$ distance in the weighted version of $G$ as a random variable.
In most of the literature, $G$ is a $d$-dimensional grid for some constant $d$ (e.g. \citep{kesten1993speed}) and $s$ and $t$ are faraway points within the grid. (This is to minimize the variance of the distance relative to its mean.) There is some work (e.g. \citep{aldous2016}) that studies first passage percolation on general graphs. However, what one can prove in general graphs is inherently limited by the presence of edges with high edge weight. We deal with this challenge in Theorem \ref{thm:main} by probing the edge weights of high-weight edges, and arguing (like in \citep{kesten1993speed}) that the remaining edges have low total variance. For more background, see \citep{kesten1987percolation, kesten1993speed, steif2009survey, auffinger201550}.

With the same setting, a more related question is the \textit{Canadian Traveler Problem} where the goal is still to find the shortest path but the edge weight is revealed once we reach one of its endpoints.
This allows some adaptive routing strategies to optimize the path \citep{papadimitriou1991shortest, nikolova2008route}, but there are still no known efficient algorithms. 
Recent works \citep{bnaya2009canadian, bhaskara2020adaptive} also allow probing in advance and predicting models to find the optimal shortest path in practice.
The setting is also closely related to \textit{Stochastic Shortest Path Problem} where we want to find the path with the minimum total expected weight. Papers are widely ranged from computing the deterministic optimal strategy \citep{bertsekas1991analysis} to an online learning setting in the context of regret minimization \citep{rosenberg2019online, rosenberg2020near, tarbouriech2021stochastic, cohen2021minimax}.

In our paper, we show a much simpler strategy to probe significant edges on a graph with special properties which are shown to be true in many real traffic networks \citep{schultes2007dynamic, geisberger2008contraction, AbrahamFG10, zhu2013shortest}.

\section{Preliminaries and Problem Setting}
\label{sec:algorithm}
\subsection{Notation}
Let $G = (V_G, E_G, w^{G})$ be a weighted undirected graph with weight $w^{G}_e$ for each edge $e$.
Let $\ell_G(P)$ be the length of path $P$ in graph $G$.
Let $P_G(u,v)$ be the shortest path from $u$ to $v$ with length $d_G(u, v)$.
Given $r > 0$ and $u \in V$, let $B^G_u(r) = \{v \in V \mid d_G(u,v) \leq r\}$
be the ball of radius $r$ centered at $u$.
We may omit the script $G$ when it is clear from the context.
\subsection{Setting: Independent Model}

In this paper, we prove theoretical results in two settings: the independent setting and the correlated setting. The results in the independent setting rely on fewer assumptions and are simpler, while the results in the correlated setting are more relevant to the problem of routing under real-world traffic conditions. We already defined the correlated setting in Section \ref{subsec:cor}. In this section, we define the independent setting and give the probing algorithm used (Algorithm \ref{alg:main_alg}). It is helpful to look at this algorithm, as the algorithm in the correlated setting (Algorithm \ref{alg:demand_alg}) is a generalization of Algorithm \ref{alg:main_alg}.

In the independent setting, we are given a weighted undirected
graph $G=(V, E, w)$, source $s$ and destination $t$. We often refer to the number of vertices $|V|$ as $n$.
Consider the following edge weight distribution, $\mathcal{D}$, where the weight of each edge $e$ is sampled from a distribution $\Omega_e$ bounded between $[w_e, \rho w_e]$ (the distribution can vary from edge to edge). 

A new graph $G' = (V, E, w')$ is a random graph obtained by re-weighting the edges of $G$ according to $w' \sim \mathcal{D}$. 
Call the new weights $w'$ \textit{the actual weights}, and the new graph $G'=(V,E,w')$ \textit{the actual graph}.
\footnote{One might wonder why we refer to the random variable as the actual value.
This is motivated by the setting when the real value is unknown to us, as we have seen in real traffic.} 
Note that only $G$ and $\mathcal{D}$ are known to us -- we do not know $w'$, but we can probe an edge $e\in E$ to learn its actual weight $w'_e$.

Our goal is to estimate the shortest path length 
from $s$ to $t$ in $G'$ using as few probes as possible -- and ideally recover such a path. For this purpose, we propose the following algorithm. We will see how to choose the parameters $c$ and $\epsilon$ later 
in the theoretical analysis of this algorithm.

\begin{algorithm}[ht]
\caption{Algorithm when edge weights are independent}
\label{alg:main_alg}
\begin{algorithmic}
    \Procedure{ApproximateLength}{$s, t, \mathcal{D}, G, G'$}
        \State \( L \gets d_G(s,t) \) 
        \State \( w'' \sim \mathcal{D} \)
        \State \(G'' \gets (V, E, w'') \)
        \State \( H \gets G''[G \cap B^G_{s}(\rho L)] \)
        \For{\( e \in E_H \)}
            \If {\( w_e > c \epsilon^2 L / (\rho^4\log n) \)}
                \State \( w^H_e \gets w'_e \)
            \EndIf
        \EndFor
        \State \Return \( d_H(s,t) \)
    \EndProcedure
\end{algorithmic}
\end{algorithm}

The main idea is to simulate our own version of $G'$ with the given information, $G$ and $\mathcal{D}$. 
We adjust some edge weights until our simulation is close enough to $G'$. 
Then, we output the shortest path length in our simulation.

To construct our version of $G'$, we first create a graph $G''$ by 
re-weighting the edges of $G$ according to $w'' \sim \mathcal{D}$. 
Our simulation $H$ is then created as an induced subgraph of $G''[G \cap B^G_s(\rho L)]$ where $L$ is the length of $s-t$ shortest path in $G$.
In other words, $H$ is the graph $G''$ which includes only nodes $v$ such that $d_{G}(s, v) \leq \rho L$.

We adjust some edge weights of $H$ by probing the actual edge weight $w_e'$ 
when $w_e$ is above a threshold $c\epsilon^2 L / (\rho^4\log n)$. 
Our main result, \cref{thm:main}, guarantees that 
the length of $s-t$ shortest path in the simulation graph $H$ 
is approximately equal to $d_{G'}(s,t)$.

While Theorem \ref{thm:main} bounds the approximation error of $\texttt{ApproximateLength}$ for all graphs, it requires an assumption in order to bound the number of queried edges. To approximate the length of the path, it suffices for the graph to have low doubling dimension.
To actually produce a path whose length is approximately shortest in $G'$, we need Assumption \ref{as:poly-paths}.

% \begin{enumerate}
%     \item Let $L$ be the length of shortest path from $s$ to $t$ in $G$.
%     \item Create a new graph $G'' = (V, E, w'')$ where $w''$ is sampled from $\mathcal{D}$.
%     \item Consider the ball of radius $\rho L$ around $s$ in $G$. Let $H$ be the 
%     induced graph of $G$ on this ball but with edge weight $w''$.
%     \item Probe all edges greater than $\frac{c \epsilon^2 L}{\log n}$ in $H$ from $G'$ and update their weights with the actual edge weights.
%     \item Return an $s$-$t$ shortest path in this ``hybrid'' graph, $H$.
% \end{enumerate}
% \medskip

\section{Theoretical Analysis}
\label{sec:theory}

In this section, we will show that our algorithm from~\Cref{sec:algorithm}
succeeds with high probability in recovering a near-optimal shortest path
for graphs with small continuous doubling dimension (\Cref{thm:main}).
\subsection{Continuous Doubling Dimension}
We first start with a definition of the doubling dimension.

\begin{definition}[Doubling Dimension \citep{AbrahamFG10}]
    A graph has doubling dimension $\cdd$ if every ball can be covered by at most $2^h$ balls of half the radius; i.e. a graph $G$ has doubling dimension $\cdd$ if for any $u \in V, r > 0$, there exists a set $S \subseteq V$ with size at most $2^\cdd$ such that $B_u(r) \subseteq \bigcup_{v \in S} B_v(r/2)$.
\end{definition}

For our results, we require a slightly different notion of doubling dimension, which we call continuous doubling dimension. Just like the difference between highway dimension and continuous highway dimension from \cite{AbrahamFG10}, in our continuous doubling dimension, we actually measure the doubling dimension of the graph after it is made continuous by subdividing all its edges into infinitesimally small segments according to their weights. The two versions of doubling dimensions do not serve as upper bound for each other, but they perform similarly in real-world networks. We use the continuous version since our proof crucially uses the fact that the union of balls in the definition must cover the entirety of each edge.

\begin{definition}[Continuous Doubling Dimension] 
\label{def:cont_cdd}
Consider a graph $G$. For a value $k$, replace each edge with
a path of length $k$ to obtain a graph $G_k$, 
where each new edge has a weight equal to $1/k$ times the original weight.
The continuous doubling dimension (cdd) of $G$ is defined to be the limit as $k$ goes to 
infinity of the doubling dimension of $G_k$.
\end{definition}
First, we will start with a simple result which says that 
in a graph with low cdd, there cannot be too many edges with large weight.
\begin{lemma}\label{lemma:few-large-edges}
    Let $G = (V, E, w)$ be a weighted graph with cdd $\cdd$. 
    If $G$ has diameter $L$, then there are at most $O\left(\left( \frac{3 L}{r}\right)^\cdd\right)$ edges with weight larger than $r > 0$.
\end{lemma}
\begin{proof}
We first show that the number of edges with a weight larger than $r$ is less than the number of balls of radius $\frac{r}{3}$ needed to cover the graph.
This follows from two facts: 
(1) by definition, the midpoint of each high-weighted edge has to be covered by at least one of these balls, 
and (2) no ball of radius $\frac{r}{3}$ can cover two such midpoints simultaneously because the distance between them is at least $r$.

To compute the number of balls of radius $\frac{r}{3}$ needed to cover a ball of radius $L$ 
we apply the cdd definition recursively.
Because $G$ has cdd $\cdd$ and is contained inside a ball of radius $L$, it can be covered 
with at most $2^{\cdd \lceil \log \frac{3 L}{r} \rceil}$ balls of radius $\frac{r}{3}$. 
This proves the lemma.
\end{proof}
\subsection{Concentration of Shortest Path Lengths}
We state a lemma that will be used to show the concentration bound in our main result. The lemma is shown for a more general case where edge weights are correlated and form (possibly joint) clusters $\{C_i\}_{i=1}^m$. Each cluster $C_i$ corresponds to an independent random variable $y_i$. Denote $\vy = (y_1, y_2,\ldots, y_m)$. An edge $e$ has original weight $w_e$ and a multiplier function $f_e$ of all the variables $\vy$, but only depends on the variables whose corresponding clusters contain $e$ (i.e. $\{y_i: e \in C_i\}$). These corresponding hidden variables are called its \emph{dependent variables}. The actual traversing time on an edge $e$ when the hidden variables being $\vy$ is $f_e(\vy)\cdot w_e$. 

The concentration bound is as follows:
% \begin{lemma}\label{lemma:small-edges-concentrated}
\begin{restatable}[]{lemma}{lemmaSmallEdgesConcentrated}
\label{lemma:small-edges-concentrated}
Given a weighted graph with $G=(V,E,\{w_e\}_{e\in E})$, a set of multiplier functions on edges $\{f_e\}_{e\in E}$, source $s$ and destination $t$, and a weight threshold $W \in \mathbb{R}_+$. There are $m$ clusters $\{C_i\}_{i=1}^m$ each with cluster weight $c_i = \sum_{e\in C_i} w_e$. Each edge $e \in E$ is included in at most $\ell$ different clusters.

Consider the random variable distribution of $\vy$, denoted by $\mathcal{D'}$. For each cluster $C_i$, if $c_i > W$, then the random variable $y_i$ has fixed value $y_i'$; otherwise, the random variable is drawn from independent distributions such that for each edge $e \in E$, the function $f_e$ is bounded between $[1, \rho']$.

Then, for two sets of random variables $\vy_1, \vy_2$ drawn independently at random from $\mathcal D'$, we have
\begin{align*}
\Pr(|\delta(\vy_1) - \delta(\vy_2)| \geq \tau)
\leq 8 \exp\left(\frac{-\tau^2}{16 \rho' (\rho'-1)^2 \cdot W\cdot \ell\cdot \sup_\vy \delta(\vy)}\right),
\end{align*}
where $\delta(\vy)$ is the length of shortest path from $s$ to $t$ in $G_\vy = (V, E, \{w_e\cdot f_e(\vy)\}_{e\in E})$.
\end{restatable}
% \end{lemma}

The proof is provided in \cref{app:proof_concentration}.

\subsection{Analysis of Algorithm: Independent Model}
Our main result for the independent model is the following.
\thmMain*
%
%
% In order to prove \Cref{thm:main}, let's first start with a bound on the number
% of queries made by algorithm from \Cref{sec:algorithm}. 
\begin{proof}
First, we need to show that 
the simulation graph $H$ does not exclude any possible shortest path.
Consider any path $P$ from $s$ with a length greater than $\rho L$ in $G$. 
Its actual length $\ell_{G'}(P)$ cannot be smaller than $\rho L$ (since $w'_e \ge w_e$ for all $e$).
On the other hand, the $s$-$t$ shortest path in $G$ 
has length at most $\rho L$ in the actual graph, $G'$, therefore $P$ cannot be 
the shortest path.
This implies that our simulation graph, $H$, contains all possible shortest paths in the actual graph $G'$.

By definition, the graph $G \cap B_s^G(\rho L)$ has a diameter no greater than $\rho L$ and we probe every edge 
with weight greater than $c \epsilon^2 L/ (\rho^4 \log n)$. 
By \Cref{lemma:few-large-edges}, the number of such edges (and thus 
the number of probes) is at most $((\rho \log n)/\epsilon)^{O(\cdd)}$. 

Note that in graph $H$, 
each edge $e$ with $w_e > W := c\epsilon^2 L / (\rho^4 \log n)$ is probed and 
so has a fixed value, whereas the remaining edges are random and can take any value in the interval $\left[w_e, \rho w_e\right]$. 
Because both edge weight of $H$ and $G'$ is drawn from the same distribution, $\mathcal{D}$ with equal large edge weight. 
Thus, we can apply \Cref{lemma:small-edges-concentrated} where each cluster is a singleton, and $\tau = \epsilon \delta$ for some appropriately chosen constant $c = 1/16$ and obtain the desired concentration bound on the shortest path length 
\begin{align*}
    Pr(|\hat{\delta} -\delta| \geq \epsilon \delta) 
\leq 8\exp\left(\frac{-\delta^2 \rho^3 \log n}{(\rho-1)^2 L \cdot \sup_z \delta(z)}\right)
\leq 8/n,
\end{align*}
where the inequality follows from $L \leq \delta$ and $\sup_z \delta(z) \leq \rho L \leq \rho \delta$.
\end{proof}

\subsection{Analysis of Algorithm: Correlated Model}
\label{sec:demand_analysis}

In \cref{alg:demand_alg}, we first conduct a normalization step. After the normalization step, we obtain a graph $\tilde{G}$ with all its edge weights $\{\tilde{w}_e\}_{e\in E}$ equal to the traversing time under basic demands. Since the original traversing time satisfies \cref{eqn:demand_constraint}, for each set of hidden random variables $\vy$, we can represent the traversing time on an edge $e$ by $w_e'(\vy)=\tilde{w}_e\cdot f_e(\vy)$, where the multiplier function $f_e(\vy)$ is bounded between $[1, \rho^\beta]$. 

We define the \emph{normalized size} of a cluster $C_i$ as its actual size under the basic demand, that is $\sum_{e\in C_i} \tilde{w}_e$ in \cref{alg:demand_alg}. To prove \cref{thm:highway-demands}, we first present a lemma where we show that a ball cannot intersect with too many clusters with large normalized cluster sizes. The proof is omitted here and can be found in \cref{app:large_clusters}. 
\begin{restatable}[]{lemma}{Clusterbounds}
    \label{lem:cluster_bounded}
    Given any graph $G$ and a set of clusters $\{C_i\}_{i=1}^m$, if the graph has continuous highway dimension $h$, each cluster is a shortest path in $G$, and each edge appears in at most $\ell$ different clusters, any $R$-radius ball centered at a point $s$ intersects with at most $h^3 \cdot \ell \cdot \left( \frac{R}{\Gamma} \right)^{\log h}$ $\Gamma$-large clusters.
\end{restatable}

Using the above lemma and the concentration bound stated in \cref{lemma:small-edges-concentrated}, we can prove \cref{thm:highway-demands} as follows.

\begin{proof} [Proof of \cref{thm:highway-demands}]
    \Cref{alg:demand_alg} probes all the hidden variable in the set: $\{y_i: \Omega_i > \Gamma'\}$, where $\Gamma' = \frac{\eps^2\cdot L}{16\cdot \log 2n\cdot \rho^{3\beta} \cdot \ell}$. By the second property in \cref{thm:highway-demands}, the total number of probes is the number of $\Gamma'$-large shortest paths that intersect with the $\rho^\beta \cdot L$-radius ball. By applying \cref{lem:cluster_bounded} with $\Gamma = \Gamma'$ and $R = \rho^{\beta} \cdot L$ on $\tilde{G}$, the number of probes is bounded by
    \begin{align*}
        h^3\cdot \ell \cdot \left(\frac{\rho^\beta \cdot L}{\Gamma'}\right)^{\log h} =\ h^3\cdot \ell \cdot \left(\frac{\rho^\beta \cdot L \cdot  16\cdot \log 2n\cdot \rho^{3\beta} \cdot \ell}{\eps^2\cdot L}\right)^{\log h}
        =\, \left(\frac{\log n\cdot \rho^{\beta} \cdot \ell}{\eps^2}\right)^{O(\log h)}.
    \end{align*}
    After the normalization step each edge in $\tilde{G}$ only varies within a multiplicative factor of $\rho^{\beta}$, we extract the subgraph with all the candidate shortest paths and denote it by $H$ in the algorithm. Since any shortest path in $\tilde{G}$ has length at most $\rho^\beta \cdot L$. $H$ has radius at most $\rho^\beta \cdot L$. Suppose the real hidden random variable is $\vy$ and \cref{alg:demand_alg} uses fake sample $\vy'$ to generate the shortest path. The real shortest path length is $\delta(\vy)$ and \cref{alg:demand_alg} returns the length $\delta(\vy')$. In \cref{alg:demand_alg} we have probed all the clusters with normalized size more than $\frac{\eps^2\cdot L}{16\cdot \log 2n\cdot \rho^{3\beta} \cdot \ell}$. Since $\vy$ and $\vy'$ can be viewed as two sets of random variables independently drawn from the same distribution, by applying \cref{lemma:small-edges-concentrated} with $W=\frac{\eps^2\cdot L}{16\cdot \log 2n\cdot \rho^{3\beta} \cdot \ell}$, $\rho' = \rho^\beta$, $\tau = \epsilon \cdot L$, graph $G = (V, E, \{\tilde{w}_{e\in E}\})$ and functions $\{f_e\}_{e\in E}$, we have $\Pr\big[|\delta(\vy) - \delta(\vy')| > \epsilon \cdot L\big] \le 1/n$. Therefore, \cref{alg:demand_alg} outputs a length which is $(1+\eps)$-approximation to the real shortest path length with high probability, completing the proof. 
\end{proof}

%In the previous section, we mathematically showed, when real-time traffic is random, that
\iffalse
\begin{enumerate}
\item The $s$-$t$ distance with real-time traffic can be approximated with a small amount of real-time traffic information. (Theorem \ref{thm:main})
\item The $s$-$t$ shortest path with real-time traffic cannot be found in some cases without a large amount of real-time traffic information. (Theorem \ref{thm:lb})
\end{enumerate}
\fi

\section{Experiments}

In this section, we observe that edge probes do indeed help estimate $s-t$ distances in $G'$ in the correlated setting. We even observe this in the special case in which each edge is affected by exactly one hidden variable.

\subsection{Experiment Setup}

In our experiments, we take a road network from Open Street Maps (OSM). In our experiments, we construct a graph, where each vertex represents a road segment and each arc represents a valid transition between a pair of road segments. We construct graphs for two different regions:

\begin{enumerate}
    \item Baden-Wurttemberg, a state in southwestern Germany
    \item Washington State, USA
\end{enumerate}

We chose both regions as they are somewhat different, large enough to have medium length trips, and small enough to fit on one machine. In both regions, we run one experiment. For each region, we generate a set of 100 queries. Each query is a single pair of points in the graph, selected uniformly at random subject to the constraint that the points are between 5 and 20 miles of one another. For context, the maximum distance between any pair of points as the crow flies in Baden-Wurttemberg and Washington is approximately 180 and 400 miles respectively.

We use simulated traffic generated via the model described in Section \ref{subsec:cor}. In our experiments, we use $\beta = 1$ for simplicity. We obtained clusters $\{C_i\}_i$ as follows:

\begin{enumerate}
    \item Consider all OSM segments that are on highways (priority 0 in OSM) or are highway exit ramps (priority 1 in OSM). Let $X$ be the set of arcs in the graph corresponding to these segments.
    \item Hash these segments to the unique S2 cell at level 8 that contains them. Define $\{C_i\}_i$ to be the resulting partition of $X$.
    \item For each arc $e$ not in $X$, add a singleton cluster $\{e\}$ to the family $\{C_i\}_i$.
\end{enumerate}

Note that each arc in the graph is in exactly one $C_i$. As in Section \ref{subsec:cor}, associate a hidden random variable $y_i$ with $C_i$, where all of the hidden random variables are chosen independently and uniformly from the interval $[1,2]$. For all arcs $e\in X$, $\lambda_i^e = 1$, where $i$ is the unique value for which $e\in C_i$. For any other pair $(e,i)$, $\lambda_i^e = 0$. For an arc $e$, $w_e$ is the travel time in seconds required to cross the arc.

We use a simplified version of Algorithm \ref{alg:demand_alg} to produce an approximate path, where the threshold is changed for simplicity. In particular, we have a \textit{threshold scale} $t_s$ which is used to adjust the threshold. 
Its main purpose is to see the effect of the number of clusters queried on path length approximation performance. 
For each region (Washington and Baden-W\"urttemberg), we enumerate the threshold scale $t_s$ within the value set $\{1 \times 10^{-5}, 1.2 \times 10^{-5}, 1.4 \times 10^{-5}, 1.6 \times 10^{-5}, 1.8 \times 10^{-5}, 2 \times 10^{-5}, 4 \times 10^{-5}, 5 \times 10^{-5}, 6 \times 10^{-5}, 7 \times 10^{-5}, 8 \times 10^{-5}, 9 \times 10^{-5}, 1 \times 10^{-4}, 1 \times 10^{-3}, 3 \times 10^{-3}, 0.01, 0.03, 0.1, 0.3, 1\}$. For each pair of points in the 100 origin-destination pairs we generate, denote by $L$ the no-traffic shortest path length between the points. We probe all the clusters with total weight above the following threshold: 
\[\Gamma = \frac{\hat{L} \cdot \eps^2}{\log n \cdot t_s},\] 
where $\hat{L} = L_{\min} \cdot 2^{\lfloor \log_2 (L/L_{\min}) \rfloor}$ is the largest power of $2$ multiplied by the minimum no-traffic path length $L_{\min}$ between any generated point pair, such that the product is no more than $L$.

Specifically, let $H$ be the graph constructed in that algorithm; that is the graph with arc weights obtained by probing all clusters with total arc weight above the threshold $\Gamma$. For the query pair $(s,t)$, let $P_H$ and $P_{G''}$ denote the $s-t$ shortest path in $H$ and $G''$ respectively, where $G''$ is an identically sampled copy of $G'$; i.e. the graph with no probes. Define the \emph{probed approximation ratio} for the query pair $(s,t)$ to be the ratio of the length of $P_H$ in the real graph $G'$ to the length of the $s-t$ shortest path in $G'$. Define the \emph{no-probe approximation ratio} for the query pair $(s,t)$ to be the ratio of the length of $P_{G''}$ in $G'$ to the length of the shortest path in $G'$.

After selecting each threshold scale, we count the fraction of probed clusters with respect to the total number of clusters in the entire graph. We pick the maximum fraction among all 100 queries and use this specific fraction as a ``probed fraction upper bound'' corresponding to that threshold scale. In \cref{fig:baden} and \cref{fig:washington}, we present the plot with this fraction upper bound of each threshold scale as the $x$-axis \footnote{We use the logarithmic scale on the $x$-axis when plotting.}, and the 90\% percentile of the probed approximation ratio as the $y$-axis.
The drastic drop in the first few nodes supports our intuition that the path length approximation performance can be improved by a few probes.

\begin{figure}
     \centering
     \begin{subfigure}[b]{0.45\textwidth}
         \centering
         \includegraphics[width=\textwidth]{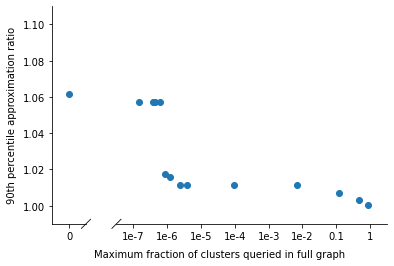}
         \caption{Baden-W\"urttemberg}
         \label{fig:baden}
     \end{subfigure}
     \hfill
     \begin{subfigure}[b]{0.45\textwidth}
         \centering
         \includegraphics[width=\textwidth]{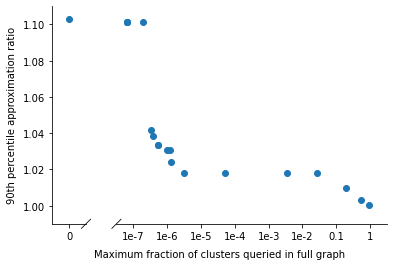}
         \caption{Washington}
         \label{fig:washington}
     \end{subfigure}
    \caption{The relationship between the maximum fraction of clusters probed among the queries and the approximation ratio obtained. Note that only a small fraction of clusters need to be probed to obtain a big improvement in approximation.}
    \label{fig:approx}
\end{figure}

% \begin{figure}
%     \centering
%     \begin{subfig}
%         \includegraphics[width=0.5\textwidth]{figures/baden_rcf_brokenaxis_new.png}
%         \caption{The probed approximation ratio of Baden-W\"urttemberg}
%         \label{fig:baden}
%     \end{subfig}
    
% \end{figure}

% \begin{figure}
%     \centering
%     \includegraphics[width=0.5\textwidth]{figures/washington_rcf_brokenaxis_new.png}
%     \caption{The probed approximation ratio of Washington}
%     \label{fig:washington}
% \end{figure}

We assessed the efficacy of probing by studying statistics of the probed and no-probed approximation ratios for the 100 queries. In Baden-W\"urttemberg,

\begin{enumerate}
    \item 90 out of 100 of all query pairs had a probed approximation ratio below 1.012 (i.e. 1.2\% distortion), with an average of .01\% of all clusters probed per query pair.
    \item the 90th out of 100 of all query pairs had a no-probe approximation ratio of 1.061 (i.e. 6.1\% distortion)
\end{enumerate}

In Washington,

\begin{enumerate}
    \item 90 out of 100 of all query pairs had a probed approximation ratio below 1.018 (i.e. 1.8\% distortion), with an average of .34\% of all clusters probed per query pair.
    \item the 90th out of 100 of all query pairs had a no-probe approximation ratio of 1.103 (i.e. 10.3\% distortion)
\end{enumerate}

Thus, in both cases, a small percentage of clusters results in a much shorter path as measured in the real traffic graph in both Washington and Baden-W\"urttemberg.

\section{Conclusion}
Path search is a fundamental problem in computer science. In many applications -- like finding driving directions in road networks -- edge weights are inherently hidden. Thus, we would like to find shortest paths with as few queries to real edge weights as possible. In this work, we modeled traffic using random, possibly correlated, edge weights and observed that we could (a) approximate the $s-t$ traffic-aware distance or (b) compute an approximate traffic-aware $s-t$ path with a small number of queries under certain realistic assumptions. Even better, (a) can be done in a small amount of runtime. Furthermore, we observed that these assumptions are fundamentally required in order to be able to find paths with a small number of queries, theoretically speaking. Experimentally, though, we observed that these results are quite pessimistic.

In future work, it would be great to turn these observations into practical data structures for answering shortest path queries with traffic using no or little preprocessing. CRP \citep{delling2011customizable} requires work to recompute shortcuts whenever edge weights in a cluster change. One could avoid that by clustering the graph into correlated pieces (highways) and applying our algorithm for correlated costs.

It would also be interesting to generalize the observations in this paper to other optimization problems. Our proof is quite simple and is not inherently tied to the shortest path problem. This could lead to simpler dynamic data structures for many problems when uncertainty in the input is random rather than adversarial.

\subsubsection*{Acknowledgements}
We would like to thank Pachara Sawettamalaya for pointing us to Talagrand's inequality in \citep{Zhao20}.

\bibliographystyle{abbrvnat}
\bibliography{cite}

%%%%%%%%%%%%%%%%%%%%%%%%%%%%%%%%%%%%%%%%%%%%%%%%%%%%%%%%%%%%
% \input{checklist}

\newpage
\appendix
\clearpage
\section{Appendix: Missing Proofs and Definitions}
In this section, we provide the missing proofs and definitions.

\subsection{Approximate Traffic Distance Data Structure}
\label{app:data_structure}
\thmDataStructure*

We now give the data structure $\texttt{ApproxTrafficDistance}$. Recall that Algorithm \ref{alg:main_alg} probes all edges in the required neighborhood with weight above a certain threshold and computes the $s-t$ distance in the graph $H$ obtained by sampling a fresh copy of $G'$ (the graph $G''$) and substituting the probed edge weights. This algorithm can be slow due to the call to Dijkstra on $H$. Instead, one can call a distance oracle on $G''$ with the probed edges deleted from the graph. This contains all of the information needed from $G''$. To compute the $s-t$ distance in $G'$, make a graph $I$ consisting of $s,t$, and the endpoints of all probed edges with non-probed edges weighted by the length of the shortest path that does not use any probed edge. $I$ has at most $O(K_0)$ vertices, so running Dijkstra on this graph is fast.

\begin{algorithm}[H]
%\caption{The data structure $\texttt{ApproxTrafficDistance}$, split into preprocessing and query methods.}
\caption{The data structure $\texttt{ApproxTrafficDistance}$ in Theorem \ref{thm:ds}.}
\label{alg:fast_alg}
\begin{algorithmic}
%    \Procedure{PreprocessATD }{$\mathcal{D}, G, G',\epsilon$}
    
 \Procedure{PreprocessApproxTrafficDistance }{$\mathcal{D}, G, G',\epsilon$}
        \State \(w'' \sim \mathcal{D}\)
        \State \(G'' \gets (V, E, w'') \)
        \State \texttt{PreprocessApproxNoTrafficDistance}($G,\epsilon/10$)
        \For{\(i\in [\log (\min_f w_f), \log (\max_f w_f)]\)}
            \State $E_i \gets e\in E(G)$  for which
            
            $w_e > c\epsilon^2 2^i / (\rho^4\log n)$
            \State \(G_i'' \gets G''\backslash E_i\)
            \State \texttt{PreprocessApproxNoTrafficDistance}($G_i'',\epsilon/10$)
            \State \(\mathcal{F}_i \gets \text{sparse cover\footnotemark for } \{B_u^G(\rho 2^i)\}_{u\in V(G)} \)
            \For{\(S\in \mathcal{F}_i\)}
                \State \(E_S \gets E_i\cap G[S]\)
            \EndFor
        \EndFor
    \EndProcedure
    \Procedure{QueryApproxTrafficDistance }{$s, t, \mathcal{D}, G, G'$}
        \State $L \gets $\texttt{QueryApproxNoTrafficDistance}($s,t,G$)
        \State $i \gets \log L$
        \State $S\gets $ a set in $\mathcal{F}_i$ that contains all of $B_s^G(100\rho 2^i)$ (exists by definition of sparse cover)
        \State $X \gets \{s,t\} \cup \{\text{endpoints of edges in $E_S$}\}$
        \State $I \gets $ complete directed graph on $X$, where $w^I_{uv} \gets $\texttt{QueryApproxNoTrafficDistance}($u,v,G_i''$) for all $u,v\in X$
        \For{\( e \in E_S \)}
            \State \( w^I_e \gets \min(w'_e, w^I_e)\)
        \EndFor
        \State \Return \( d_I(s,t) \)
    \EndProcedure
\end{algorithmic}
\end{algorithm}
\footnotetext{Theorem 3.1 \citep{AP90}, $k = \log n$.}

\begin{proof}[Proof of Theorem \ref{thm:ds}]

We first bound the runtime of the algorithm. For preprocessing, computing $w''$ and $G''$ takes $O(m)$ time and preprocessing the no-traffic data structure on $G$ takes $O(K_1)$ time. Computing all $E_i$s and $G_i''$s takes $\tilde{O}(m)$ time, as this is only done $\log((\max_f w_f) / (\min_f w_f))$ times, for a total of $O(m \log((\max_f w_f) / (\min_f w_f))) = \tilde{O}(m)$ time. Computing each $\mathcal{F}_i$ takes $\tilde{O}(m)$ time by \citep{MPX13}, for a total of $\tilde{O}(m)$ time. Preprocessing the $G_i''$ data structures takes $\tilde{O}(K_1)$ time. Computing all of the sets $E_S$ takes $\tilde{O}(m)$ by the degree property of the sparse cover; specifically computing $E_S$ takes $O(|E(G[S])|)$ time, so the total work for $\mathcal{F}_i$ is

\begin{align*}
\sum_{S\in \mathcal{F}_i} O(|E(G[S])|) &= \sum_{S\in \mathcal{F}_i}\sum_{e\in E(G[S])} O(1)\\
&=\sum_{e\in E(G)} O(\text{number of $S\in \mathcal{F}_i$ with $e\in E(G[S])$})\\
&\leq\sum_{e\in E(G)} O\left((\log n) n^{1/(\log n)}\right)\\
&= \tilde{O}(m).
\end{align*}

This completes the preprocessing time bound. For query time, computing $L$ takes $K_2$ time, $i$ takes $O(1)$ time, and $S$ takes $O(1)$ time, as only a pointer to $S$ needs to be stored. Since $|E_S| \le K_0$ by the radius bound of the sparse cover, constructing $I$ takes $\tilde{O}(K_0^2 K_2)$ time. Substituting the probed edge weights takes $K_0$ time and running Dijkstra in $I$ takes $\tilde{O}(K_0^2)$ time. Thus, the total query time is $\tilde{O}(K_0^2 K_2)$, as desired.

Now, we bound the approximation error of the returned number. By Theorem \ref{alg:main_alg}, the shortest path in $G_i''$ with probed edges from $G'$ added is a $(1 + \epsilon/10)$-approximation to the shortest path in $G'$. Let $v_0 = s, v_1, v_2, \hdots, v_{k-1}, v_k = t$ denote the subsequence of vertices visited by the path that are also in $I$. If $\{v_i, v_{i+1}\}\in E_S$, then the probed value from $G'$ is present in $I$. If $\{v_i, v_{i+1}\}\notin E_S$, then the subpath between $v_i$ and $v_{i+1}$ does not use any probed edge, so the $v_i-v_{i+1}$ subpath is also a shortest path in $G_i''$. This means that $w_{v_iv_{i+1}}^I$ is a $(1 + \epsilon/10)$-approximation to the length of the subpath, so the shortest path in $I$ is a $(1 + \epsilon/10)^2 < (1 + \epsilon)$-approximation to the length of the shortest path in $G'$, as desired.

\end{proof}

\subsection{Definition of Highway Dimension}
We use $P(v,w)$ to denote the shortest path between a pair of vertices $v$ and $w$. The definition of highway dimension is as follows:
\label{app:highway}
\begin{definition*}[Highway Dimension \citep{AbrahamFG10}]
Given a graph $G=(V, E)$, the highway dimension of $G$ is the smallest integer $h$ such that 
\begin{align*}
    &\forall\, r \in \mathbb{R}^+, \forall u\in V, \exists S\subseteq B_{u, 4r}, |S|\le h, \text{ such that }\\ &\forall\, v, w\in B_{u, 4r}, \text{ if } |P(v,w)|>r \text{ and } P(v,w)\subseteq B_{u, 4r},\\
    &\quad \text{ then } P(v,w) \cap S \neq \emptyset.
\end{align*}
\end{definition*}
Similar as the continuous double dimension, we also define the continuous highway dimension of a graph by chopping each edge into infinitely many smaller segments:
\begin{definition*}[Continuous Highway Dimension]
Consider a graph $G$. For a value $k$, replace each edge with a path of length $k$ to obtain a graph $G_k$, where each new edge has a weight equal to $1/k$ times the original weight. The continuous highway dimension of $G$ is defined to be the limit as $k$ goes to infinity of the highway dimension of $G_k$.
\end{definition*}

The continuous highway dimension can be used to upper bound the continuous doubling dimension (see \cref{def:cont_cdd}), by the following lemma:
\begin{lemma} [Upper bound of continuous doubling dimension]
\label{lem:cdd_chd}
    If a graph $G$ has continuous doubling dimension $\alpha$ and continuous highway dimension $h$, we have $2^\alpha \le h$. 
\end{lemma}
\begin{proof}
    Omitted. See Claim 1 in \citep{AbrahamFG10} for the proof.
\end{proof}
%%%%%%%%%%%%%%%%%%%%%%%%%%%%%%%%%%%%%%%%%%%%%%%%%%%%%%%%%%%%
\subsection{Algorithm of Probing Demands} \label{app:alg}

See \cref{alg:demand_alg} for the algorithm. 

\begin{algorithm}[H]
\caption{Algorithm when edge weights come from flows/can be correlated}
\label{alg:demand_alg}
\begin{algorithmic}
    \Procedure{ProbingDemands \newline}{$s, t, G, \{\mathcal{D}_i\}_{i \in [m]}, \{\lambda_i\}_{i\in [m]}, \{\lambda_i^e\}_{i \in [m], e \in E}, \vy$}
        \For{\( i \in [m], e \in E \)}
            % \State \( \tilde{\lambda}_{i}^e \gets \frac{\lambda_i \cdot \lambda_i^e}{\sum_{j \in [m]} \left(\lambda_j^e \cdot \lambda_j \right)}; \) 
            \State \( \tilde{w}_e \gets w_e \cdot \left(\sum_{i \in [m]} \left( \lambda_i^e \cdot \lambda_i \right) \right)^\beta \) 
        \EndFor
        \State \( \tilde{G} \gets (V, E, \tilde{w}) \) \Comment{Compute the edge weights under basic demands}
        \State \( L \gets d_{\tilde{G}}(s,t)\)
        \State \( H \gets \tilde{G}[\tilde{G} \cap B^{\tilde{G}}_{s}(\rho^\beta \cdot L)] \)
        \For{\( i \in [m] \)}
            \State \( \Omega_i \gets \sum_{e\in E_H \cap C_i} \tilde{w}_e \) \Comment{Compute the cluster size of demand $y_i$}
            \If {\( \Omega_i > \frac{\eps^2\cdot L}{16\cdot \log 2n\cdot \rho^{3\beta} \cdot \ell} \)}
                \State \( y'_i \gets y_i \) \Comment{Probe on $y_i$ if the effective size is large}
            \Else
                \State \( y'_i \sim \mathcal{D}_i\) \Comment{Sample from $\mathcal{D}_i$ otherwise}
            \EndIf
        \EndFor
        \For{\( e \in E_H \)}
            \State \( w_e^H \gets\left(\sum_{i=1}^{m} \lambda_{i}^e \cdot y_i'\right)^\beta \cdot w_e \)
        \EndFor        
        \State \Return \( d_H(s,t) \)
    \EndProcedure
\end{algorithmic}
\end{algorithm}

%%%%%%%%%%%%%%%%%%%%%%%%%%%%%%%%%%%%%%%%%%%%%%%%%%%%%%%%%%%%
\subsection{Concentration of Shortest Path Lengths} \label{app:proof_concentration}
Here, we prove \cref{lemma:small-edges-concentrated} that is used to establish our main results, restated as follows: 
\lemmaSmallEdgesConcentrated*

\begin{proof}
For our proof, we will need the following version of 
Talagrand's concentration inequality.
\begin{theorem}[Theorem 9.4.14 of \citep{Zhao20}]\label{thm:talagrand-use}
Let $\Omega = \Omega_1 \times \cdots \times \Omega_k$ equipped with the product measure.
Let $f: \Omega \to \mathbb{R}$ be a function.
Suppose for every $x \in \Omega$, there is some $\alpha(x) \in \mathbb{R}_{\geq 0}^k$
such that for every $y \in \Omega$,
\[f(x) \leq f(y) + d_{\alpha(x)}(x, y).\]
where $d_{\alpha(x)}(x,y) = \sum \alpha_i 1_{x_i \neq y_i}$ is the weighted Hamming distance. 
Then, for every $t \geq 0$,
\[\Pr(|f - \mathbb{M}f| \geq t) \leq 4 \exp \left(\frac{-t^2}{4 \sup_{x \in \Omega}|\alpha(x)|^2}\right)\]
where $\mathbb{M}X$ is the median for the random variable $X$; i.e.,
$\Pr(X \geq \mathbb{M}X) \geq 1/2$ and $\Pr(X \leq \mathbb{M}X) \leq 1/2$.
\end{theorem}
 
    We will apply \Cref{thm:talagrand-use} to show that the length of the shortest path is concentrated around its median. Here, $\Omega$ is the joint distribution of hidden variables $\mathcal{D}'$ 
    and $\Omega_i$ is the distribution of the hidden variable $y_i$. 
    
    We show that the condition for Theorem~\ref{thm:talagrand-use} holds for $f := -\delta$ with the following. 
    For any weight vector $\vy$ and $\vx$ drawn from $\mathcal{D}'$, let $p_\vy$ (resp. $p_\vx$) be the shortest path from $s$ to $t$ in $G_\vy = (V, E, \{w_e\cdot f_e(\vy)\}_{e\in E})$ (resp. $G_\vx = (V, E, \{w_e\cdot f_e(\vx)\}_{e\in E})$) and define $\alpha(\vx)_i = (\rho' - 1) \cdot \sum_{e \in p_\vx \cap C_i} w_e$ if $c_i < W$; and $0$ otherwise. We have 
    \begin{align*}
    \delta(\vy) 
    \leq \sum_{e \in p_\vx} w_e\cdot f_e(\vy) 
    = \sum_{e \in p_\vx} w_e\cdot f_e(\vx) + \sum_{e \in p_\vx} w_e\cdot (f_e(\vy) - f_e(\vx)).
    \end{align*}
    We partition the set $\{e: e\in p_\vx\}$ into two parts. The first part $E_1$ consists the edges whose dependent variables remain the same as $\vy$, i.e. $E_1=\{e \in p_\vx: \forall i, e\in C_i \rightarrow x_i=y_i\}$. The rest edges are in $E_2$: $E_2=\{e \in p_\vx: \exists\,i, e\in C_i, x_i \neq y_i\}$. We have
    \begin{align*}
    \delta(\vy)  & = \delta(\vx) + \sum_{e \in E_1} w_e\cdot (f_e(\vy) - f_e(\vx))  + \quad\quad \sum_{e\in E_2} w_e\cdot (f_e(\vy) - f_e(\vx))\\
    & = \delta(\vx) + \sum_{e \in p_\vx: \exists\,i, e\in C_i, x_i \neq y_i} w_e\cdot (f_e(\vy) - f_e(\vx))\\ 
    & \leq \delta(\vx) + \sum_{i: x_i \neq y_i} \sum_{e\in p_\vx \cap C_i} w_e \cdot (\rho' - 1)\\
    & = \delta(\vx) + d_{\alpha(\vx)}(\vx,\vy).
\end{align*}

Observe that
\begin{align*}
    \sup_{\vx}|\alpha(\vx)|^2 &= \sup_{\vx}\sum_
{i=1}^m{|\alpha(\vx)_i|^2} \\ &\le \sup_{\vx}\left(\sum_
{i=1}^m{|\alpha(\vx)_i|}\right) \cdot \sup_{\vx,i} \alpha(\vx)_i \\
&\le \rho' \cdot \sup_{\vx} \delta(\vx) \cdot \ell \cdot (\rho'-1) \cdot \sup_{\vx,i} \alpha(\vx)_i \\
&\le \rho' \cdot \sup_{\vx} \delta(\vx) \cdot \ell \cdot (\rho'-1)^2 \cdot W.
\end{align*}

Applying \Cref{thm:talagrand-use}, we have 
\begin{equation*}
\Pr(|\delta(x) - \mathbb{M} \delta| \geq \tau) \leq 4 \exp{\left(\frac{-\tau^2}{4 \rho' (\rho' - 1)^2 \cdot W\cdot \ell \cdot \sup_\vz \delta(\vz)}\right)}.
\end{equation*}
The theorem then follows by plugging in $\tau/2$ and applying a union bound.
\end{proof}

\subsection{Bounds on the Intersected Large Clusters}
\label{app:large_clusters}
\Clusterbounds*
\begin{proof}
    We first prove that an $r$-radius ball can only intersect with $h^2\cdot \ell$ different clusters with total weight in the range $(r,\ \infty)$. If a path is intersecting the $r$-radius ball and has total length larger than $r$, it must contain a subset of edges with total weight greater than $r$ that form a shortest path inside the $4r$-radius ball centered at $s$. By the definition of the highway dimension, the shortest path inside the $4r$-radius ball contains at least one of the highway points. By the third property, we know that each highway point inside an edge (not vertex) is covered by at most $\ell$ different clusters. A highway point on a vertex $v$ falls on at most $\ell\cdot d(v)$ different clusters, where $d(v)$ is the degree of the vertex $v$. Denote the continuous doubling dimension of the graph by $\cdd$, we have $\cdd \le \log h$ (by \cref{lem:cdd_chd}). Since we have $2^{\cdd}$ as an upper bound of the degree of a vertex, we have $d(v)\le h$. Therefore, the $r$ radius ball will intersect with at most $h^2\cdot \ell$ clusters with size in the range $(r,\ \infty)$. 
 
    Since the continuous doubling dimension of $G$ is bounded by $\log h$, we can cover the $R$-radius ball with $h^{\lceil \log R/r \rceil}$ $r$-radius balls. Let $r = \Gamma$. Each $r$-radius ball intersects with at most $h^2\cdot \ell$ $\Gamma$-large clusters, the total number of $\Gamma$-large clusters intersecting with the $R$-radius ball is at bounded by $h^{\lceil \log R/\Gamma \rceil} \cdot h ^2\cdot \ell \le h^3 \cdot \ell \cdot \left( \frac{R}{\Gamma} \right)^{\log h}$. 
\end{proof}

% ###############################################

\subsection{Finding a Short Path in \texorpdfstring{$G'$}{G'} Requires Lots of Probes}
\label{app:find_hard}

In \cref{sec:theory}, we showed that the $s-t$ distance in the real graph $G'$ can be approximated using a small number of probes to edge weights in $G'$. Even better, these probes are done \emph{non-adaptively}; i.e. the edges are probed in one batch. One may wonder whether it is possible to always produce a \textit{short path} in $G'$. 
By short path, we mean the path whose length is approximately equal to the shortest path between $s$ and $t$.
We show that this is indeed impossible, even with a large number of adaptively chosen probes. In order to formally describe the result, we first need to define adaptive probing strategies:

\begin{definition}[Adaptive probing strategies]
Consider a hidden graph $G'$. An \emph{adaptive probing strategy} is an algorithm $\mathcal A$ that takes a pair $(s,t)$ of vertices in $G'$ along with the unweighted edges of $G'$ and outputs an $s-t$ path $P$. The algorithm is also given the weights of some edges in $G'$, given as follows. The algorithm picks a sequence of edges $e_1, e_2, \hdots, e_k$ and sees their edge weights $w_1, w_2, \hdots, w_k$ respectively in $G'$. The choice of the $e_i$s is allowed to be adaptive, in the sense that the choice of $e_i$ is a function of $s,t$, and $w_1, w_2, \hdots, w_{i-1}$. 

The \emph{query complexity} of the algorithm $\mathcal A$ is the number $k$. The \emph{quality} of the path $P$, denoted $q(P)$, is the ratio $q(P) := \ell_{G'}(P) / d_{G'}(s,t)$, where $\ell_{G'}(P)$ denotes the length of $P$ in $G'$.

\end{definition}

\vspace{-0.1in}
\begin{figure}[h]
\includegraphics[width=0.5\textwidth]{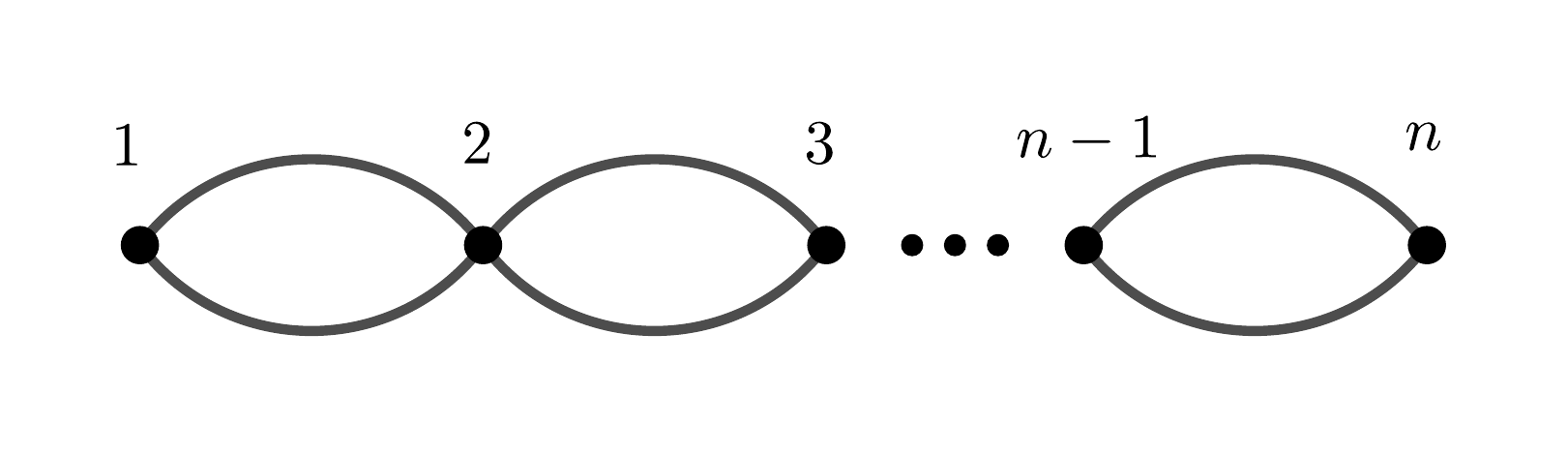}
\centering
\caption{High-query example for finding a path in $G'$}
\label{fig:counterexample}
\end{figure}

In the following proof, we set $\rho$ to be $2$ for simplicity.
We can scale the edge weight easily by considering $1 + (w_e - 1)\cdot(\rho - 1)$ instead, and the same result still follows.

We now define an example graph in which it is hard to find a short path using an adaptive probing strategy. This example is depicted in Figure \ref{fig:counterexample}. Make an $n$-vertex graph $G$ with vertices $v_1,v_2,\hdots,v_n$. Between any two consecutive vertices $v_i,v_{i+1}$, there are two edges $e_i$ and $f_i$. A hidden graph $G'$ is generated by giving edge weights to each of the edges in $G$. For $e_i$ and $f_i$, the edge weights are denoted $u_i$ and $l_i$ respectively. The $u_i$s and $l_i$s are sampled uniformly and independently from the interval $[1,2]$ as usual. When $G$ and $G'$ are used in this subsection, they will always refer to the graphs $G$ and $G'$ defined in this paragraph.

Now, consider an adaptive probing strategy $\mathcal A$. We will ask it to compute a path between $s = v_1$ and $t = v_n$. Note the following:

\begin{proposition}\label{prop:distance}
    $d_{G'}(s,t) = \sum_{i=1}^{n-1} \min(u_i, l_i)$.
\end{proposition}

\begin{proof}
Any path from $s$ to $t$ uses exactly one edge from the set $\{e_i, f_i\}$ for each $i$. Therefore, $d_{G'}(s,t) \ge \sum_{i=1}^{n-1} \min(u_i, l_i)$. The minimum weight edges, though, also yield a path from $s$ to $t$, so $d_{G'}(s,t) \le \sum_{i=1}^{n-1} \min(u_i, l_i)$ as well, as desired.
\end{proof}

Throughout our analysis, we use the classic Chernoff-Hoeffding bound:

% \begin{theorem}[Theorem 3.1 of \citep{chunglu2006}]\label{thm:chernoff}
%     Let $X_1,\hdots,X_n$ be independent random variables with $|X_i - \mathbb{E}[X_i]|\le 1$ for all $i$. Let $X = \sum_{i=1}^n X_i$ and let $\sigma^2$ be the variance of $X$. Then,

%     $$\Pr[|X - \mathbb{E}[X]| \ge k\sigma] \le 2e^{-k^2\sigma^2/4n}$$

%     for any $0\le k\le 2\sigma$.
% \end{theorem}

\begin{theorem}[Hoeffding's inequality]\label{thm:hoeffding}
    Let $X_1, \dots, X_n$ be independent random variables such that $a_i \leq X_i \leq b_i$ for all $i$. Let $X = \sum_{i = 1}^n X_i$. Then, 
    \[\Pr(|X - \mathbb{E}[X]| \ge t) \le 2 \exp\left(-\frac{2t^2}{\sum_{i=1}^n (b_i - a_i)^2}\right)\]
    for all $t > 0.$
\end{theorem}

We use Chernoff to upper bound the denominator of the quality value:

\begin{proposition}\label{prop:distance-ub}
    With probability at least $1 - 1/n^{100}$, $d_{G'}(s,t) \le 4n/3 + n/100$.
\end{proposition}

\begin{proof}
    We apply Chernoff with $X_i = \min(u_i, l_i) \in [1, 2]$. Thus, by Theorem \ref{thm:hoeffding},

    $$\Pr\big(|X - \mathbb{E}[X]| \ge n/100\big) \le 2 e^{-\Theta(n)} \le n^{-100}.$$

    Furthermore, $\mathbb{E}[X_i] = 4/3$ for all $i$, because

    \begin{align*}
        \mathbb{E}[X_i] &= \int_1^2 (\Pr[X_i > x] + 1)dx\\
        &= 1 + \int_1^2 \Pr(u_i > x)\Pr(l_i > x)dx\\
        &= 1 + \int_1^2 (2-x)^2dx\\
        &= 4/3.
    \end{align*}

    Thus, $\mathbb{E}[X] \le 4n/3$, which means that $d_{G'}(s,t) = X \le 4n/3 + n/100$ with probability at least $1 - 1/n^{100}$ as desired. 
\end{proof}

Thus, to get a lower bound on the quality (approximation ratio), we just need to lower bound the numerator. We show the following:

\begin{lemma}\label{lem:length-lb}
Any adaptive probing strategy with query complexity at most $n/100$ outputs a path $P$ with $\ell_{G'}(P) \ge 3n/2 - n/20$ with probability at least $1 - 1/n^{100}$.
\end{lemma}

\begin{proof}
Let $\mathcal A$ be an adaptive probing strategy with query complexity $k$ for some $k\le n/100$. The probed values $w_1,w_2,\hdots,w_k$ are random variables. Consider some fixing of these random variables. This fixing induces a fixed choice of edges $g_1,g_2,\hdots,g_k$ queried by the algorithm and a choice of one path $P$. Let $S\subseteq \{1,2,\hdots,n-1\}$ be the minimum set of indices for which $g_1,g_2,\hdots,g_k\subseteq \cup_{i\in S} \{e_i, f_i\}$. In particular, $|S|\le k$ and for every $i\notin S$, neither $l_i$ nor $u_i$ were queried by $\mathcal A$.

For any $i$, let $h_i$ be the single edge among $\{e_i, f_i\}$ that the path $P$ uses and let $X_i$ be the weight of $h_i$ in $G'$ (either $l_i$ or $u_i$). The $X_i$s are random variables. Conditioned on $w_1,w_2,\hdots,w_k$, the $X_i$s for $i\notin S$ are independent because the choice of path $P$ is not a function of $l_i$ or $u_i$. Thus, we may apply Chernoff to lower bound their sum. By Theorem \ref{thm:hoeffding},

\begin{equation*}
\Pr\Bigg(\bigg|\sum_{i\notin S} X_i - \mathbb{E}\Big[\sum_{i\notin S} X_i \;|\; w_1,\hdots,w_k\Big]\bigg|> n/50 \;\Big|\; w_1,\hdots,w_k\Bigg) \le 2e^{-\Theta(n)} < n^{-100}.
\end{equation*}

Furthermore, for any $i\notin S$,

$$\mathbb{E}[X_i | w_1,\hdots,w_k] = 3/2$$

and because $|S|\le k\le n/100$,

$$\mathbb{E}\left[\sum_{i\notin S} X_i | w_1,\hdots,w_k \right] = (3/2)(n - |S| - 1) > 3n/2 - n/50.$$

Combining these statements shows that

$$\Pr\left(\sum_{i\notin S} X_i \le 3n/2 - n/20 \Big| w_1,\hdots,w_k\right) \le n^{-100}.$$

By the tower law of conditional expectations,

\begin{align*}
\Pr\left(\sum_{i\notin S} X_i \le 3n/2 - n/20\right) &= \mathbb{E}\left[\Pr\left(\sum_{i\notin S} X_i \le 3n/2 - n/20 \Big| w_1,\hdots,w_k\right)\right]\\
&\le \mathbb{E}[n^{-100}]\\
&\le n^{-100}.
\end{align*}

By definition, it is always the case that $\ell_{G'}(P) \ge \sum_{i\notin S} X_i$, so $\ell_{G'}(P) > 3n/2 - n/20$ with probability at least $1 - n^{-100}$, as desired.
\end{proof}

We are now ready to show the main lower bound:

\thmLB*

\begin{proof}
    Lemma \ref{lem:length-lb} shows that $\ell(P) \ge 3n/2 - n/20$ with probability at least $1 - n^{-100}$ over the choice of edge weights in $G'$. Proposition \ref{prop:distance-ub} shows that $d_{G'}(s,t) \le 4n/3 + n/100$ with probability at least $1 - n^{-100}$ over the choice of $G'$. Thus, by a union bound,

    $$q(P) \ge (3n/2 - n/20) / (4n/3 - n/100) > 9/8 - 1/10 > 1$$

    with probability at least $1 - 2n^{-100}$, as desired.
\end{proof}

\subsection{Getting around the lower bound} 
\label{app:get_around}

There is something very unrealistic about the Figure \ref{fig:counterexample} example, though. Specifically, there are exponentially many possible shortest paths from $s$ to $t$. This does not make intuitive sense in road networks -- such paths could only arise if a car exited and re-entered a highway a large number of times. Thus, the following assumption makes sense:

\asPoly*

Mathematically, for any origin-destination pair $s-t$, 
let $S$ be a set of all possible graph $\tilde{G} = (V, E, \tilde{w})$ constructed by sampling $\tilde{w} \sim \mathcal{D}$, then the size of set of possible shortest paths $\{P_{\tilde{G}}(s,t) | \tilde{G} \in S\}$ is polynomial in $|V|$. 

\thmPathAlgo*
In the analysis, we make use of a one-sided McDiarmid's Inequality:

\begin{theorem}[Theorem 6.1 of \citep{chunglu2006}, with martingale given by sum of first $i$ variables and applied on negation to get two sidedness]\label{thm:mcdiarmid}
Let $X_1,\hdots,X_n$ be independent random variables with $|X_i - \mathbb{E}[X_i]|\le M$ for all $i$. Let $X = \sum_{i=1}^n X_i$. Then, for any $\lambda > 0$,

$$\Pr[|X - \mathbb{E}[X]| \ge \lambda] \le 2e^{-\frac{\lambda^2}{\sum_{i=1}^n \text{Var}(X_i) + M\lambda/3}}.$$
\end{theorem}

\begin{proof}
    By the edge count bound of Theorem \ref{thm:main}, the algorithm only probes $((\rho \log nU)/\epsilon)^{O(h)}$ edges, so it suffices to show that the algorithm finds an $(1 + \epsilon)$-approximate path in $G'$. By the polynomial paths assumption, the shortest $s-t$ path in $G'$ is one of $U$ different paths $P_1,P_2,\hdots,P_U$. We start by showing that, with probability at least $1 - 4/((nU)^{100})$,
    $$(1 - \epsilon)\ell_{G'}(P_i) \le \ell_H(P_i) \le (1 + \epsilon) \ell_{G'}(P_i)$$
    for each $i$. To do this, think about the construction of $H$ slightly differently. Think of the construction of $H$ as replacing the low edge weights in $G'$ with edge weights in $G''$. In particular, all large edge weights are deterministic, and all small ones are randomized. Break the edges of $P_i$ into two sets $A$ and $B$, depending on whether their weight in $G'$ (within a factor of $\rho$ of the same value for $G$) is greater than or less than $c\epsilon^2 L / (\rho^4 \log nU)$ respectively. Recall that $w_e'$ and $w_e''$ denote the weights of $e$ in $G'$ and $G''$ respectively. We now use Theorem \ref{thm:mcdiarmid} to bound the error. Note that for $e\in B$,

    \begin{align*}
    \text{Var}(w_e'') &\le \mathbb{E}[(w_e'')^2]\\
    &\le \rho^2(w_e')^2\\
    &\le \rho^2(w_e')\max_{f\in B} w_f'\\
    &\le w_e' \cdot c\epsilon^2 L /(\rho^2 \log nU).\\
    \end{align*}

    Therefore,
    \begin{align*}
        \sum_{e\in B} \text{Var}(w_e'') &\le \ell_{G'}(P_i)\cdot c\epsilon^2 L/(\rho^2 \log nU)\\ &\le c \epsilon^2 L^2/(\rho \log nU)
    \end{align*}

    since $\rho L \ge \ell_{G'}(P_i)$ (otherwise $P_i$ cannot be a candidate shortest path). Using $M = c\epsilon^2 L/(\rho^4 \log nU)$, $X_e = w_e''$ for all $e\in B$ and $X_e = w_e'$ for all $e\in A$, and $\lambda = \epsilon L/2$ shows that

    $$\Pr[|X - \mathbb{E}[X]| \ge \epsilon L/2] \le 2/(nU)^{100}$$

    where $X = \sum_{e\in A} w_e' + \sum_{e\in B} w_e''$, which is also the length of $P_i$. This inequality also uses the fact that $c < 1/800$.\footnote{Recall that we set $c$ to be $1/16$ in the proof of \cref{thm:main}; however, it is true for any $c \leq 1/16$, so we can set it to be less than $1/800$ and every result still remains true.} 
    Notice that $G'$ and $H$ are both samples from the distribution over edge weights that this probability bound pertains to. Thus,

    $$\Pr[|\ell_H(P_i) - \mathbb{E}[X]| \ge \epsilon L/2] \le 2/(nU)^{100}$$
    and
    $$\Pr[|\ell_{G'}(P_i) - \mathbb{E}[X]| \ge \epsilon L/2] \le 2/(nU)^{100}.$$

    By Union bound,
    $$\Pr[|\ell_{G'}(P_i) - \ell_H(P_i)| \ge \epsilon \ell_{G'}(P_i)] \le 4/(nU)^{100}$$
    because $L \leq \ell_{G'}(P_i)$.
    This is the first desired probability bound. 
    Now, we discuss how to use it to prove the theorem. Union bound over all $U$ paths to show that all of these inequalities hold simultaneously with probability at least $1 - n^{-100}$. 
    Let $P$ denote the shortest path in $G'$. 
    Applying the inequalities to $P_i = P$ shows that $P$ is only a $(1 + \epsilon)$-factor longer in $H$. 
    This means that the path $Q$ that the algorithm returns has length at most $(1 + \epsilon)\ell_{G'}(P)$. 
    $Q$ is one of the $P_i$s, because it is the shortest path in $H$, which is a valid sample from the distribution that $G'$ is sampled. 
    Thus, we may apply the inequality to it to show that the length of $Q$ in $G'$ is at most $\frac{1 + \epsilon}{1-\epsilon}\cdot\ell_{G'}(P) \leq (1 + 3\epsilon)\ell_{G'}(P)$. 
    Thus, $Q$ is a $(1 + 3\epsilon)$-approximate path in $G'$ as desired.
\end{proof}

\end{document}